\documentclass{llncs}
\usepackage{epsfig}
\usepackage{algorithm}
\usepackage{amsmath}
\usepackage{amssymb}
\usepackage{bm}
\usepackage{alltt}

\newcommand{\tuple}[1]{\left<{#1}\right>}

\begin{document}

\title{Note on the Lower Bounds of Bimachines}
\author{Stefan Gerdjikov\inst{1,2}}

\institute{
Faculty of Mathematics and Informatics\\
Sofia University\\
5, James Borchier Blvd., Sofia 1164, Bulgaria\\
stefangerdzhikov@fmi.uni-sofia.bg
\and
Institute of Information and Communication Technologies\\
Bulgarian Academy of Sciences\\
25A, Acad. G. Bonchev Str., Sofia 1113, Bulgaria\\
stoyan@lml.bas.bg
}

\maketitle


\begin{abstract}
 This is a brief note on the lower bound of bimachines. Particularly, we report that
 there is a class of functional transducers with $O(n)$ states that do not admit a bimachine
 with fewer than $\Theta(2^n)$ states. 
\end{abstract}

\section{Introduction}
Finite state transducers and bimachines~\cite{Eil74,Sakarovitch09} are formalisms that describe exactly the class of rational functions. However, whereas the transducers  
bear the nondeterminism of the nondeterministic finite state automata, the bimachines provide a deterministic, strictly linear procedure to process a given input.
Thus, a natural question arises what the cost of the determinism is in terms of space. The classical construction,~\cite{Sch61,Berstel79,RS97}, given a functional transducer with $n$ states produces a bimachine with $O(n!)$ states. In~\cite{TCSSubmitted2017} we provided a construction that results in a bimachine with $O(2^n)$ states. In this report we show that there are hard instances of transducers with $O(n)$ states that require a bimachine with at least $2^n$ states.

\section{Transducers and Bimachines}\label{SecFormalPreliminaries}
For basic notions on alphabets and automata we refer to~\cite{Sakarovitch09,Eil74}. 

A \emph{finite state transducer} is ${\cal T}=\tuple{\Sigma\times\Omega^*,Q,I,F,\Delta}$ where $\Sigma$ and $\Omega$ are alphabets,
$Q$ is a final set of \emph{states}, $I,F\subseteq Q$ are the \emph{initial} and \emph{final} states, respectively, and $\Delta\subseteq Q\times(\Sigma\times \Omega^*)\times Q$
is a finite relation of \emph{transitions}.

Similarly to automata, it is instructive to consider a transducer as a graph with labels. The semantics of transducers is defined in terms of \emph{paths} in this graph.
A path in a transducer ${\cal T}=\tuple{\Sigma\times\Omega^*,Q,I,F,\Delta}$ is either (i) a non-empty sequence of transitions of the form:
\begin{equation*}
\pi= \tuple{p_0,\tuple{a_1,\omega_1},p_1}\dots \tuple{p_{n-1},\tuple{a_n,\omega_n},p_n} \text{ with } \tuple{p_i,\tuple{a_{i+1},\omega_{i+1}},p_{i+1}}\in \Delta,
\end{equation*} 
or (ii) $\pi=(p)$ where $p\in Q$. Each path $\pi$ features a \emph{source} state, $\sigma(\pi)$, a \emph{terminal} state, $\tau(\pi)$, a \emph{label}, $\lambda(\pi)$, and \emph{length}, $|\pi|$. These terms are defined as follows:
\begin{eqnarray*}
\text{case (i) } &\sigma(\pi)=p_0;\quad \tau(\pi)=p_n;\quad &\lambda(\pi)=\tuple{a_1\dots a_n,\beta_1\dots\beta_n}; \quad |\pi|=n\\
\text{case (ii) } & \sigma(\pi)=p;\quad \tau(\pi)=p;\quad &\lambda(\pi)=\tuple{\varepsilon,\varepsilon}; \quad |\pi|=0.\\
\end{eqnarray*}
With these notions, the language, i.e. \emph{rational relation}, recognised by a transducer ${\cal T}$ is:
\begin{equation*}
{\cal R}({\cal T})=\{\lambda(\pi)\,|\, \pi \text{ is a path with } \sigma(\pi)\in I \text{ and } \tau(\pi)\in F\}.
\end{equation*}

A transducer ${\cal T}$ is called \emph{functional} if ${\cal R}({\cal T})$ is a graph of a function $f_{\cal T}:\Sigma^*\rightarrow \Omega^*$.

\begin{definition}\label{DefBimachine}
A {\em bimachine} is a tuple ${\cal B}=\tuple{{\cal M},{\cal A}_L,{\cal A}_R,\psi}$, where:
\begin{itemize}
\item ${\cal A}_L = \tuple{\Sigma,L,s_L,L,\delta_L}$ and ${\cal A}_R = \tuple{\Sigma,R,s_R,R,\delta_R}$ are deterministic finite-state automata.
\item ${\cal M} = \tuple{M,\circ,e}$ is the output monoid and $\psi : (L \times \Sigma \times R) \rightarrow M$ is a partial function.
\end{itemize}
Note that all states of ${\cal A}_L$ and ${\cal A}_L$ are final. The function $\psi$ is naturally extended to the  {\em generalized output function}\index{Generalized output function of a bimachine} $\psi^\ast$ as follows:
\begin{itemize}
\item $\psi^\ast(l,\varepsilon,r)=e$ for all $l\in L, r\in R$;
\item $\psi^\ast(l,t\sigma,r)= \psi^\ast(l,t,\delta_R(r,\sigma)) \circ  \psi(\delta^\ast_L(l,t),\sigma,r)$ for $l\in L, r\in R, t\in\Sigma^\ast, \sigma\in\Sigma$.
\end{itemize}
The {\em function represented by the bimachine} is 
$$O_{\cal B}:\Sigma^\ast \rightarrow M: t  \mapsto \psi^\ast(s_L,t,s_R).$$
\end{definition}

 \section{Lower Bound}
Whereas the functionality property of transducers can be algorithmically tested,~\cite{BCPS03,Sch61}, they are non-deterministic devices and thus the processing of
the input depends on the size of the transducer. On the other hand the bimachines represent a deterministic device that also recognises functions. It turns out that both formalisms have the same expressive power,~\cite{Sch61,RS97}. However the classical constructions,~\cite{Sch61,RS97}, feature a worst case $2^{\Theta(n\log n)}$ blow-up 
in terms of states. Recently, in~\cite{TCSSubmitted2017} was reported a construction leading to $O(2^{|Q|})$ states and thus improving on the $2^{\Theta(n\log n)}$ bound.

In this section we shall address the problem of determining the lower bound of an algorithm that constructs a bimachine.
 For a positive integer $k>1$, we consider the alphabet $\Sigma_k=\{1,2,\dots,2k\}$ with $2k$ distinct characters.
For each $k$ we are going to construct a family of transducers ${\cal T}_n^{(k)}$ with $2k(n+1)$ states that do not admit an equivalent bimachine with less than $k^n +1$ states. In this sense we obtain a lower worst-case bound of $2^{\Theta(n)}$ for any construction that transforms a functional transducer into an equivalent bimachine.

Our approach is based on the classical example of regular languages that are hard to determinise. In particular, we the underlying
hard instances of rational functions are encoded in languages similar in spirit to $\{0,1\}^* 1\{0,1\}^{n-1}$.
  
Formally, for $1\le i\le k$ and $k+1\le j\le 2k$, we introduce the language:
\begin{equation*}
L^{(i,j)}_{n}= \{1,2,\dots, k\}^* i \{1,2,\dots, k\}^{n-1}\{k+1,\dots ,2k\}^{n-1}j \{k+1,\dots ,2k\}^*.
\end{equation*}
In words this is the language that has a block of characters from the first half of the alphabet followed by a block
with characters from the second half of the alphabet such that the last but
$n$-th character in the first block is $i$ and the $n$-th character in the second block is $j$.
Next we define the function $f^{(k)}_n:\Sigma_k^*\rightarrow \Sigma_k^*$ with domain the union of languages $L^{(i,j)}_n$:
\begin{eqnarray*}
f^{(k)}_n(\alpha) = \begin{cases}
ji \text{ if } \alpha \in L_n^{(i,j)} \text{ for } i\in \{1,2,\dots,k\}\text{ and } j\in \{k+1,\dots,2k\},\\
\neg ! \text{ otherwise.}
\end{cases}
\end{eqnarray*}

\begin{lemma}\label{lemma:transducer}
For each integer $k>0$ and $n\in \mathbb{N}$ there is a transducer with $2k(n+1)$ states representing $f^{(k)}_n$.
\end{lemma}
\begin{proof}[Sketch] Consider the regular expressions:
Let $\Sigma_k'=\{1,2,\dots,k\}$ and $\Sigma_k''=\{k+1,\dots,2k\}$, then:
\begin{eqnarray*}
R'_i &=&(\Sigma_k'\times\{\varepsilon\})^*\tuple{i,\varepsilon}(\Sigma_k'\times\{\varepsilon\})^{n-1} \text{ for } i\in \Sigma_k' \text{ and }\\
R''_i &=&(\Sigma_k''\times\{\varepsilon\})^{n-1}\tuple{i,\varepsilon}(\Sigma_k'\times\{\varepsilon\})^* \text{ for } j\in \Sigma_k''.
\end{eqnarray*}
\begin{figure}
\includegraphics[width=1\textwidth]{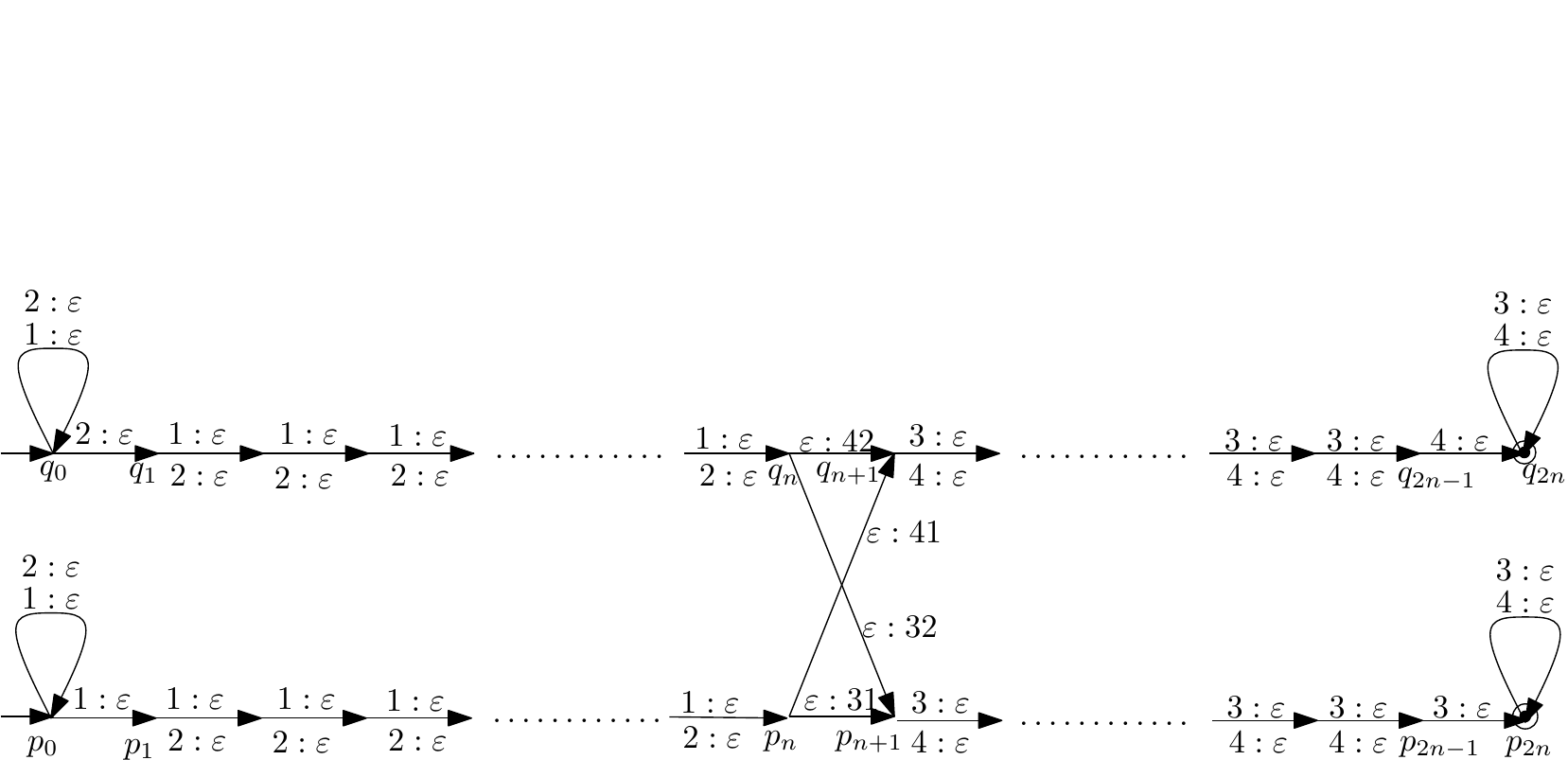}
\caption{A transducer representing the function $f^{(2)}_n$.}\label{fig:LowerBound}
\end{figure}
For each of them we can construct a transducer with $n+1$ states amounting to total of $2k(n+1)$ states, see Figure~\ref{fig:LowerBound}. Furthermore, it is straightforward that:
\begin{equation*}
f^{(k)}_n = \bigcup_{i\in\Sigma_k'}\bigcup_{j\in\Sigma_k''} R'_i \{\tuple{\varepsilon,ji}\} R''_j.
\end{equation*}
We can easily express these operation by laying $\{\tuple{\varepsilon,ji}\}$ transitions between the final state of $R'_i$ and the initial state of $R''_j$.

\end{proof}

\begin{remark}\label{remark:simple}
One can easily merge the initial states so that instead of $k$ initial states the transducer has only $1$. Similarly, one can
collapse the final states in a single one. This results in a transducer with total $2kn+2$ states.
\end{remark}

\begin{lemma}\label{lemma:bimachine}
For $k\ge 2$, every bimachine representing $f^{(k)}_n$ has at least $k^n+1$ states in total.
\end{lemma}
\begin{proof}
Let ${\cal B}=\tuple{\Sigma^*,{\cal A}_L,{\cal A}_R,\psi}$ be a bimachine such that ${\cal O}_{\cal B}=f_n$.
Let ${\cal A}_L=\tuple{\Sigma,Q_L,s_L,Q_L,\delta_L}$ and ${\cal A}_R=\tuple{\Sigma,Q_R,s_R,Q_R,\delta_R}$.
We prove that $|Q_L|\ge k^n$ or $|Q_R|\ge k^n$.

Assume that both inequalities failed, thus $|Q_L|<k^n$ and $|Q_R|<k^n$. Hence, there are distinct words $\alpha',\alpha''\in \{1,2,\dots,k\}^n$
such that:
\begin{equation*}
p'_L=\delta_L^*(s_L,\alpha')=\delta_L^*(s_L,\alpha'').
\end{equation*}
Without loss of generality we may assume that $\alpha'=\alpha 1 x'$ and $\alpha''=\alpha 2 x''$. Let $x\in \{1,2,\dots,k\}^*$ be of length $|\alpha|$. 
Let $\alpha_1=\alpha 1 x'x$ and $\alpha_2 = \alpha 2 x'' x$. Note that $|x'x|=|x''x|=n-1$. We define $p_L$ as:
\begin{equation*}
p_L=\delta_L^*(p'_L,x).
\end{equation*}
Let $\gamma\in \{k+1,\dots,2k\}^*$ be of length $|\gamma|=m\ge n$. Let $\gamma=c_1c_2\dots c_m$ and set:
\begin{equation*}
p_L^{(l)}=\delta_L^*(p_L,c_1\dots c_l) \text{ and } q_R^{(l)} =\delta_R^*(s_R,c_m\dots c_{l+1}).
\end{equation*} 
We claim that $\Psi(p_L^{(l)},c_{l+1},q_R^{(l)})=\varepsilon$. Indeed, since $\alpha_i\in \{1,2,\dots,k\}^* i \{1,2,\dots,k\}^{n-1}$ and $|\gamma|\ge n$
we have that $\alpha_i\gamma$ is in the domain of $f^{(k)}_n$. Therefore $\Psi(p_L^{(l)},c_{l+1},q_R^{(l)})$ are defined for each $l$.
If $\Psi(p_L^{(l)},c_{l+1},q_R^{(l)})\neq \varepsilon$ for some $l$, then we can consider the last $l'$ such that $\Psi(p_L^{(l')},c_{l'+1},q_R^{(l')})\neq \varepsilon$.
Thus, the last character of $\Psi(p_L^{(l')},c_{l'+1},q_R^{(l')})$ is the last character of both: ${\cal O}_{\cal B}(\alpha_1\gamma)$ and ${\cal O}_{\cal B}(\alpha_2\gamma)$.
Since ${\cal O}_{\cal B}(\alpha_1\gamma)=f^{(k)}_n(\alpha_1\gamma)$ its last character is $1$. Similarly, the last character of ${\cal O}_{\cal B}(\alpha_2\gamma)=f_n^{(k)}(\alpha_2\gamma)$ is $2$. Contradiction!

Applying a symmetric argument with respect to ${\cal A}_R$ we construct words $\beta_3=y' a\beta$ and $\beta_4=y'' b \beta$ of length $n$ such that $a,b\in \{k+1,\dots,2k\}$ are distinct and further:
\begin{equation*}
p_R' =\delta_R^*(s_R,\beta_3^{rev})=\delta_R^*(s_R,\beta_4^{rev}).
\end{equation*} 
We pick an arbitrary word $y\in \{k+1,\dots,2k\}^*$ of length $|y|=|\beta|$ and define:
\begin{equation*}
p_R =\delta_R^*(p'_R,y^{rev}).
\end{equation*} 
As above we see that for every word $a_1\dots a_m\in \{1,2,\dots,k\}^*$ of length $m\ge n$ setting:
\begin{equation*}
q_L^{(l)}=\delta_L^*(s_L,a_1\dots a_l) \text{ and } p_R^{(l)}=\delta_R^*(p_R,a_{l+1}\dots a_m)
\end{equation*}
 it holds that $\Psi(q_L^{(l)},a_{l+1},p_R^{(l)})=\varepsilon$.
 
Applying this observation for $\alpha_1\beta_3$ we obtain that ${\cal O}_{\cal{B}}(\alpha_1\beta_3)=\varepsilon$. However $f^{(k)}_n(\alpha_1\beta_3)=a1$, a contradiction.

Therefore, $|Q_L|\ge k^n$ or $|Q_R|\ge k^n$. Since $|Q_L|\ge 1$ and $|Q_R|\ge 1$, the result follows. \qed
\end{proof}

\begin{corollary}
There is a family of rational functions $f_n$ that can be represented with $|Q_n|$ states but every bimachine representing $f_n$
requires at least $2^{c (|Q_n|-2)}$ where $c=\frac{\log_2 3}{6}$.
\end{corollary}
\begin{proof}
In view of Lemma~\ref{lemma:transducer}, Remark~\ref{remark:simple}, and Lemma~\ref{lemma:bimachine} we have that for
every $k\ge 2$, the functions $f_n^{(k)}$ can be represented by a transducer with $|Q_n|=2kn+2$ states and requires at least $k^n+1$ in terms
of bimachines. Expressing:
\begin{equation*}
k^n+1=2^{n\log_2 k }+1=2^{\frac{|Q_n| - 2}{2k}\log_2 k} +1
\end{equation*}
and setting $c_k=\frac{\log_2 k}{2k}$ we get that each bimachine representing $f_n^{(k)}$ requires at least $2^{c_k (|Q_n|-2)}$ states where $|Q_n|$
is the size of the minimal transducer representing $f_n^{(k)}$. It is easy to see that $c_k$ is maximised for $k=3$ and in this case we obtain:
\begin{equation*}
c_3=\frac{\log_2 3}{6}.
\end{equation*} 
\end{proof}
\section{Conclusion}\label{SecConclusion}

In this paper we considered the lower bound problem for bimachine construction. We showed that every construction of a bimachine
out of a transducer features a worst case state requirements $O(2^{n \frac{\log_2 3}{6}})$. This comes close to the construction $O(2^n)$
that we have from~\cite{TCSSubmitted2017}. Yet, it remains open whether we can improve the construction of a bimachine, say to $O(2^{n/2})$,
or there are still harder instances for bimachine construction that show that $O(2^n)$ is tight.

\bibliographystyle{splncs03}
\bibliography{bibliography-2}

\end{document}